\documentclass[11p,reqno]{amsart}

\usepackage[T1]{fontenc}
\usepackage{graphicx}
\usepackage{amssymb,amsthm,amsmath,mathrsfs,bm,braket,marginnote}
\usepackage{enumerate}
\usepackage{appendix}
\usepackage[colorlinks=true, pdfstartview=FitV, linkcolor=blue, citecolor=blue, urlcolor=blue]{hyperref}
\usepackage{multirow}

\usepackage{pgf}
\usepackage{pgfplots}
\usepackage{tikz}
\usetikzlibrary{arrows,calc}
\usepackage{verbatim}
\usetikzlibrary{decorations.pathreplacing,decorations.pathmorphing}
\usepackage[numbers,sort&compress]{natbib}
\usepackage{dsfont}

\numberwithin{equation}{section}
\newtheorem{theorem}{Theorem}[section]
\newtheorem{lemma}[theorem]{Lemma}

\newtheorem{coro}[theorem]{Corollary}

\newtheorem{proposition}[theorem]{Proposition}

\theoremstyle{remark}
\newtheorem{remark}[theorem]{Remark}

\begin{document}

\title[]{Relations between generalised Wishart matrices, the Muttalib--Borodin model and matrix spherical functions}


\date{}

\author{Peter J. Forrester}
\address{School of Mathematical and Statistics, The University of Melbourne, Victoria 3010, Australia}
\email{pjforr@unimelb.edu.au}

\dedicatory{}


\begin{abstract}
Generalised uncorrelated Wishart matrices are formed out of rectangular standard Gaussian data matrices with a certain pattern of zero entries. Development of the theory in the real and complex cases has proceeded along separate line. For example, emphasis in the real case has been placed on the Bellman and Riesz distributions, while that in the complex case has been shown to be closely related to the Muttalib-Borodin model. In this work, as well as uniting the lines of development, a tie in with matrix spherical functions is identified in the context of deducing the eigenvalue probability density function from the joint element probability density function.
\end{abstract}

\maketitle

\section{Introduction}
\subsection{Some classes of real and complex random matrix ensembles}
The purpose of this article is to link two strands of random matrix theory which hitherto have been considered in isolation. Chronologically, the first of the strands concerns the probability density function (PDF) on the space of real positive definite matrices,
\begin{equation}\label{qW1a}
{1 \over \mathcal N_N^{(1)}(\boldsymbol \alpha)} q^{(1)}_{\boldsymbol \alpha}({W}^{(1)}) e^{- {\rm Tr} \, {W}^{(1)}/2}\chi_{{W}^{(1)}>0},
\end{equation}
where
\begin{equation}\label{1.7a}
q^{(1)}_{\boldsymbol \alpha}({W}^{(1)}) :=
\prod_{i=1}^{N-1} ( \det ({W}_i^{(1)}))^{(\alpha_i - \alpha_{i+1} - 1)/2}
( \det {W}^{(1)} )^{(\alpha_N-1)/2}
\end{equation}
and
\begin{align}\label{q1W}
\mathcal N_N^{(1)}(\boldsymbol \alpha) := 2^{\sum_{j=1}^N (\alpha_j+1)/2} (2\pi)^{N(N-1)/4} \prod_{i=1}^N \Gamma((\alpha_i + 1)/2). 
\end{align}
In (\ref{1.7a}) ${W}_i^{(1)}$ denotes the leading $i \times i$ sub-matrix of ${W}^{(1)}$ and the parameters $\{ \alpha_i \}$ are such that $\alpha_i > -1$ ($i=1,\dots,N$). Generally we write $\chi_A = 1$ for $A$ true, and $\chi_A = 0$ otherwise, while $W^{(1)} > 0$ denotes the requirement that $W^{(1)}$ be positive definite.

Work of Veleva \cite{Ve09a,Ve09a,Ve11} (see also \cite{GL19,GKLS22}) has given an entry-wise construction of random matrices $Y^{(1)}$ such that $W^{(1)}= (Y^{(1)})^T Y^{(1)}$ has joint entry PDF (\ref{qW1a}).

\begin{proposition}\label{P1x} (Veleva)
Let $Y^{(1)}$ be an upper triangular random matrix with all entries real. The $k$-th diagonal entry is required to be such that its square has distribution $\Gamma[(\alpha_k +1)/2,{1/2}]$ with $ \alpha_k > -1$, while the strictly upper triangular entries of $Y^{(1)}$ are independent standard real Gaussians. The joint element PDF of $W^{(1)}= (Y^{(1)})^T Y^{(1)}$ is then equal to (\ref{qW1a}). Alternatively, with $n$ a positive integer, $n \ge N$, now require that the parameters $\{ \alpha_i \}$ be non-negative integers obeying the ordering
\begin{equation}\label{1.1x}
1 + \alpha_1 \le 2 + \alpha_2 \le \cdots \le N + \alpha_N \le n.
\end{equation}
Introduce the $n \times N$ rectangular random matrices $\tilde{Y}^{(1)}$ by the
requirement first that there is a certain pattern of zero entries,
\begin{equation}\label{1.1bx}
\tilde{Y}_{jk}^{(1)} = 0 \qquad {\rm if} \: \: j > k + \alpha_k.
\end{equation}
(Note that such prescribed zero entries then occur in the strictly lower triangular portion only.) Let the entries of $\tilde{Y}^{(1)}$ not prescribed to equal zero by (\ref{1.1bx}) be independent standard real Gaussians. One has that the positive definite random matrix $\tilde{W}^{(1)} := (\tilde{Y}^{(1)})^T \tilde{Y}^{(1)}$ also has the joint element PDF (\ref{qW1a}).
\end{proposition}

\begin{remark}  $ $ \\
1.~In the case $\alpha_i = n - i$ Proposition \ref{P1x} is classical. Thus all entries of $\tilde{Y}^{(1)}$ are then independent standard real Gaussians, and so the product $\tilde{W}^{(1)} = (\tilde{Y}^{(1)})^T \tilde{Y}^{(1)}$ is  by definition an uncorrelated real Wishart matrix (see e.g.~\cite[Def.~3.2.1]{Fo10}). The joint element distribution of $\tilde{Y}^{(1)}$ in this case is proportional to $e^{-{\rm Tr} \, (\tilde{Y}^{1)})^T\tilde{Y}^{(1)})/2}$. Pioneering work of Wishart \cite{Wi28} gave that the joint element PDF of $\tilde{W}^{(1)}$ is proportional to $(\det \tilde{W}^{(1)} )^{(n-N-1)/2} e^{-{\rm Tr} (\tilde{W}^{(1)})/2}$ in keeping with the specialisation of (\ref{qW1a}). Now apply the Gram-Schmidt algorithm to the columns of $Y^{(1)}$ to obtain the QR decomposition $Y^{(1)} = Q^{(1)}R^{(1)}$, where the columns of $Q^{(1)}$ form an orthonormal basis for the column space of $Y^{(1)}$, and $R^{(1)}$ is an $N \times N$ upper triangular real matrix with diagonal entries strictly positive. One sees that $\tilde{W}^{(1)} = (R^{(1)})^T R^{(1)}$. Moreover, it was shown by Barlett \cite{Ba33} that the matrix $R^{(1)}$ is specified as for the case $\alpha_i = n - i$ in the definition of the upper triangular matrix $Y^{(1)}$ given in the statement of the proposition.\\
2.~That (\ref{qW1a}) is a PDF follows from work of Bellman \cite{Be56} on a generalisation of the matrix gamma function introduced by Ingham \cite{In33} and Siegel \cite{Si35}; for more on this viewpoint see \cite{GN00}. Earlier still, albeit less directly in that the space of positive definite matrices must be viewed as a particular Jordan algebra, the structure (\ref{qW1a} can be found in \cite{Ri49}. See too \cite[Pg.~137]{FK94} and in particular the dedicated text \cite{Ha21}.
\end{remark}

In the theory of Wishart matrices (see e.g.~\cite{Mu82}), the columns of the matrix $\tilde{Y}^{(1)}$ as defined in Proposition \ref{P1x} can be regarded as containing centered data. The zeros as required by (\ref{1.1bx}) then have the interpretation as missing data \cite{Ve09a}. Up to proportionality, $\tilde{W}^{(1)}$ is the covariance matrix for the measurements. Adopting this viewpoint, there is also interest in the correlation matrix
\begin{equation}\label{C1x}
\tilde{C}^{(1)} := \tilde{A}^{-1/2} \tilde{W}^{(1)} A^{-1/2}, \qquad
\tilde{A} = {\rm diag} \, (\tilde{W}^{(1)}_{11},\tilde{W}^{(1)}_{22},\dots,\tilde{W}^{(1)}_{NN}).
\end{equation}
The explicit form of the joint element PDF is known from \cite{Ve09b}.

\begin{proposition}\label{P1y} (Veleva)
Define $\tilde{C}^{(1)}$ in terms of $\tilde{W}^{(1)}$ as in (\ref{C1x}), or more generally define 
\begin{equation}\label{C1y}
{C}^{(1)} := A^{-1/2} {W}^{(1)} A^{-1/2}, \qquad
A = {\rm diag} \, ({W}^{(1)}_{11},{W}^{(1)}_{22},\dots,{W}^{(1)}_{NN}),
\end{equation}
with ${W}^{(1)}$ as in Proposition \ref{P1x}. The joint element PDF of $C^{(1)}$ is given by
\begin{equation}\label{qW1A}
{\prod_{l=1}^N \Gamma((\alpha_l + l)/2) \over \mathcal N_N^{(1)}(\boldsymbol \alpha)} q^{(1)}_{\boldsymbol \alpha}({C}^{(1)}) \chi_{C^{(1)}_{ii} =1 \,(i=1,\dots,N), \,{C}^{(1)}>0}.
\end{equation}

\end{proposition}

\begin{remark}
In the case $\alpha_i = n - i$, when $\tilde{W}^{(1)}$ corresponds to an uncorrelated real Wishart matrix, the result (\ref{qW1A}) is due to Muirhead \cite[Th.~5.1.3]{Mu82}.
\end{remark}

We turn our attention now to the second strand in the theme of our study. This has its origins in the class of eigenvalue PDFs on the space of complex positive definite matrices
\begin{equation}\label{1.0a}
 {1 \over \prod_{l=1}^N l! \,\Gamma(\theta (l-1) + c +1)} \prod_{l=1}^N  x_l^c e^{-x_l} \chi_{x_l > 0} \prod_{1 \le j < k \le N} (x_j - x_k) (x_j^\theta - x_k^\theta).
\end{equation}
This specifies the Muttalib--Borodin model of random matrix, first introduced in the study of quantum transport by Muttalib \cite{Mu95}, and subsequently shown to permit exact evaluation of its limiting hard edge $k$-point correlation by Borodin \cite{Bo98}, who isolated a biorthogonal structure. A tie in with the triangular random matrices $Y^{(1)}$ of Proposition \ref{P1x} shows itself upon recalling a result of Cheliotis \cite{Ch14} on the construction of random matrices realising (\ref{1.0a}) for the eigenvalue PDF.

\begin{proposition}\label{P1.2x}
Let $Y^{(2)} = [y_{j,k}]_{j,k=1}^N$ be an upper triangular random matrix with diagonal entries distributed by the requirement that $|y_{k,k}|^2 \mathop{=}\limits^{\rm d} \Gamma[\alpha_k+1,1]$ on the diagonal, and with strictly upper triangular entries distributed as standard complex Gaussians. Introduce the positive definite random matrices $W^{(2)} = (Y^{(2)})^\dagger Y^{(2)}$ and let $\{x_j\}_{j=1}^N$ denote the corresponding eigenvalues. For the choice
\begin{equation}\label{1.0a+}
\alpha_k = \theta(k-1) + c, \qquad \theta \ge 0, \: \: c > - 1,
\end{equation}
the eigenvalue PDF of $W^{(2)}$ is given by (\ref{1.0a}).
\end{proposition}

Forrester and Wang \cite{FW17} extended the computation of the eigenvalue PDF of $W^{(2)}$ to all $\alpha_k > -1$ in the definition of $Y^{(2)}$. Moreover, in the case that the $\{ \alpha_k \}$ are non-negative integers satisfying the inequalities (\ref{1.1x}), random matrices analogous to $\tilde{Y}^{(1)}$ in Proposition \ref{P1x} were also identified.

\begin{proposition}\label{P1.3x}
(Forrester and Wang)
Construct the random upper triangular $Y^{(2)}$ as in Proposition \ref{P1.2x}, and similarly $W^{(2)}$. The eigenvalue PDF of the positive definite random matrix $W^{(2)}$ is
\begin{equation}\label{1.0}
{1 \over N!} {1 \over \prod_{l=1}^N \Gamma(\alpha_l+1)} {1 \over \prod_{1 \le j < k \le N} (\alpha_j - \alpha_k)} \prod_{l=1}^N e^{-x_l} \chi_{x_l > 0} \prod_{1 \le j < k \le N} (x_j - x_k) \det [ x_j^{\alpha_k} ]_{j,k=1}^N.
\end{equation}
Furthermore, with the parameters $\{ \alpha_k \}$ non-negative integers and required to satisfy the inequalities (\ref{1.1x}),
introduce particular $n \times N$ rectangular random matrices $\tilde{Y}^{(2)}$. These are specified by the
requirements first that there is a pattern of zero entries,
\begin{equation}\label{1.1b}
\tilde{Y}_{jk}^{(2)} = 0 \qquad {\rm if} \: \: j > k + \alpha_k,
\end{equation}
as is identical to (\ref{1.1bx}), and second 
that the entries of $\tilde{Y}^{(2)}$ not prescribed to equal zero by (\ref{1.1b}) be independent standard complex Gaussians. One has  that the positive definite random matrix $\tilde{W}^{(2)} := (\tilde{Y}^{(2)})^\dagger \tilde{Y}^{(2)}$ has the same joint element PDF as  $W^{(2)}$, and in particular the eigenvalue PDF of $\tilde{W}^{(2)}$ is given by (\ref{1.0}). 
\end{proposition}

\begin{remark}\label{R1.7}
In the case $\alpha_l = n - l$ all entries of $\tilde{Y}^{(2)}$ are non-zero and distributed as independent standard complex Gaussians. Then $\tilde{W}^{(2)}$ is by definition an uncorrelated complex Wishart matrix (see e.g.~\cite[Def.~3.2.1]{Fo10}). The corresponding eigenvalue PDF is then given by the functional form proportional to
(see e.g.~\cite[Prop.~3.2.2 with $\beta = 2$]{Fo10})
\begin{equation}\label{1.1c+}
\prod_{l=1}^N x_l^{n-N} e^{-x_l} \chi_{x_l > 0}
\prod_{1 \le j < k \le N}(x_k - x_j)^2.
\end{equation}
This is consistent with the corresponding specialisation of (\ref{1.0}) upon recalling the Vandermonde determinant identity
\begin{equation}\label{1.1d+}
\det [x_j^{k-1}]_{j,k=1,\dots,N} = \prod_{1 \le j < k \le N} (x_k - x_j).
\end{equation}
One notes too that in this case the eigenvalues of $\tilde{W}^{(2)}$ are the squared singular values of the rectangular complex Gaussian matrix $\tilde{Y}^{(2)}$ (also known as a rectangular GinUE matrix; for more on this viewpoint see \cite[\S 6.3]{BF22a}).
\end{remark}

Extending the pioneering work of Borodin \cite{Bo98}, there is now an extensive body of literature on the statistical properties of the eigenvalues specified by (\ref{1.0a}). As a non-exhaustive list we reference \cite{CR14,Zh15, FL15,FW17,FI18,La18,CGS19,CC21,CLM21,Mo21,FL22,WZ22}. Of these we highlight two results which we will show below have companions for the analogous cases of the matrices $W^{(1)}$ and $\tilde{W}^{(1)}$. The first actually applies to the full parameter range of the random matrices $\tilde{W}^{(2)}$ as specified in Proposition \ref{P1.3x}. Thus for such random matrices, we have from \cite[Corollary 3.7]{FI18} that the averaged characteristic polynomial has the explicit form
\begin{equation}\label{1.1e+}
\Big \langle \det(x\mathbb I_N - \tilde{W}^{(2)}) \Big \rangle =
\sum_{\nu=0}^N x^\nu \sum_{k=0}^\nu {(-1)^{N-k} \over (\nu - k)! k!}
\prod_{l=1}^N (\alpha_l + k + 1).
\end{equation}
In relation to the second, which is particular to (\ref{1.0a}) and thus the specialisation of parameters (\ref{1.0a+}), for $\theta \in \mathbb R^+$ introduce the sequence
\begin{equation}\label{1.1f+}
C_\theta(k) = {1 \over \theta k + 1}
\binom{\theta k + k}{k}, \qquad k=0,1,2,\dots
\end{equation}
This sequence, for $\theta \in \mathbb Z^+$, forms an integer sequence referred to as the Fuss-Catalan sequence. Moreover, for general $\theta>0$, $\{C_\theta(k)\}$ are known to be moments of a PDF, referred to as the Fuss-Catalan density, and moreover they uniquely determine the PDF \cite{Ml10,BBCC11}. The point of interest is that upon the change of variables $\lambda_l \mapsto (\lambda_l/(N \theta)^\theta)^\theta$, the normalised eigenvalue density of (\ref{1.0a}) is given by the Fuss-Catalan density; see \cite[Prop.~3.1]{FW17}, which interprets results from \cite{CR14}.

\subsection{Extension of the theory}\label{S1.1}

Comparison of the above account of the development of the two strands of random matrix theory that form out theme 
makes it evident that aspects in the theory of the real case relative to the complex case are incomplete, and vice versa. In particular, missing from the literature are results for

\begin{enumerate}
    \item The joint element PDF of $W^{(2)}, \tilde{W}^{(2)}$.
    \item The joint element PDF of the correlation matrices corresponding to $W^{(2)}, \tilde{W}^{(2)}$.
    \item The eigenvalue PDF of $W^{(1)}, \tilde{W}^{(1)}$.
    \item Evaluation of the characteristic polynomial 
    and  the normalised eigenvalue density
    for $\tilde{W}^{(1)}$.
\end{enumerate}

Our main achievements in this paper are to extend the existing theory by specifying the missing results.

\subsection*{Result (1)}
Our first result is for the functional form of the joint element PDF of $W^{(2)}$, which from the final statement of Proposition \ref{P1.3x} includes that for $\tilde{W}^{(2)}$.

\begin{proposition}\label{P1}
Define
$$
\mathcal N_N^{(2)}(\boldsymbol \alpha) := \pi^{N(N-1)/2} \prod_{i=1}^N \Gamma(\alpha_i + 1)
$$
and
\begin{equation}\label{1.17b}
q^{(2)}_{\boldsymbol \alpha}({W}^{(2)}) :=
\prod_{i=1}^{N-1} ( \det ({W}_i^{(2)}))^{\alpha_i - \alpha_{i+1} - 1}
( \det {W}^{(2)} )^{\alpha_N},
\end{equation}
where ${W}_i^{(2)}$ denotes the leading $i \times i$ sub-matrix of ${W}^{(2)}$. One has that the joint element PDF of ${W}^{(2)}$ is equal to
\begin{equation}\label{1.5a}
{1 \over \mathcal N_N^{(2)}(\boldsymbol \alpha)} q^{(2)}_{\boldsymbol \alpha}({W}^{(2)}) e^{- {\rm Tr} \, {W}^{(2)}}\chi_{{W}^{(2)}>0}.
\end{equation}
\end{proposition}

\begin{remark}
Consider the matrix $\tilde{W}^{(2)} = (\tilde{Y}^{(2)})^\dagger \tilde{Y}^{(2)}$ in the case $\alpha_l = n - l$. As noted in Remark \ref{R1.7}, all entries of the $n \times N$ random matrix $\tilde{Y}^{(2)}$ are independent standard complex Gaussians, and $\tilde{W}^{(2)}$ is an uncorrelated complex Wishart matrix.
 A fundamental result in the theory of Wishart matrices specifies the joint element PDF of $\tilde{W}^{(2)}$ as proportional to \cite[Prop.~3.2.7 with $\beta = 2$]{Fo10}
\begin{equation}\label{1.1c}
(\det \tilde{W}^{(2)})^{n-N} e^{- {\rm Tr} \, \tilde{W}^{(2)}} 
\chi_{\tilde{W}^{(2)}>0},
\end{equation}
which is in agreement with the result implied by Proposition \ref{P1}.
\end{remark}

\subsection*{Result (2)}
Our second result is for the functional form of the joint element PDF of the correlation matrix corresponding to $W^{(2)}$.

\begin{proposition}\label{P1y+} 
Define 
\begin{equation}\label{C1r}
{C}^{(2)} := B^{-1/2} {W}^{(2)} B^{-1/2}, \qquad
B = {\rm diag} \, ({W}^{(2)}_{11},{W}^{(2)}_{22},\dots,{W}^{(2)}_{NN}),
\end{equation}
with ${W}^{(2)}$ as in Proposition \ref{P1.3x}. The joint element PDF of $C^{(1)}$ is given by
\begin{equation}\label{qW1C}
{\prod_{l=1}^N \Gamma(\alpha_l + l) \over \mathcal N_N^{(2)}(\boldsymbol \alpha)} q^{(2)}_{\boldsymbol \alpha}({C}^{(2)}) \chi_{C^{(2)}_{ii} =1 \,(i=1,\dots,N), \,{C}^{(2)}>0}.
\end{equation}

\end{proposition}

\begin{remark}
In the case $\alpha_i = n - i$ the result (\ref{qW1A}) can be found in \cite[Remark 9.2]{FZ20}.
\end{remark}

\subsection*{Result (3)}
We turn our attention now to the eigenvalue PDF implied by the joint element PDF (\ref{qW1a}) for $W^{(1)}$. To obtain an explicit functional form, we must restrict attention to $\tilde{W}^{(1)}$.
Then $\{\alpha_k\}$ are distinct non-negative integers ordered from $\alpha_1$ the smallest to $\alpha_N$ the largest as is compatible with (\ref{1.1x}). We set
$\alpha_k = \gamma_{N+1-k} + k - 1$ so that by (\ref{1.1x}) the array
\begin{equation}\label{gat}
\boldsymbol \gamma = (\gamma_1,\dots, \gamma_N)
\end{equation}
is such that the entries are weakly decreasing and so $\boldsymbol \gamma$ forms a partition.

\begin{proposition}\label{P1.4}
    Let $\{ \alpha_k \}$ be distinct non-negative integers, and define $\boldsymbol \gamma$ as in (\ref{gat}) and surrounding text. Impose the further condition that the parts of $\boldsymbol \gamma$ all have the same parity.
    Let $\mathcal Y_{\boldsymbol \kappa}(x_1,\dots,x_N)$ denote the zonal polynomial of mathematical statistics (see \cite[Ch.~7]{Mu82}). For this case of the joint element PDF (\ref{qW1a}), the eigenvalue PDF is
equal to
\begin{equation}\label{qW1c+}
c_N { 2^{-\sum_{j=1}^N(\gamma_j - 1)/2}
\over \prod_{j=1}^N \Gamma((\alpha_j+1)/2)}
{\mathcal Y_{(\boldsymbol \gamma - 1 + s)/2}(x_1,\dots,x_N) 
\over 
\mathcal Y_{(\boldsymbol \gamma - 1 + s)/2}((1)^N)}
\prod_{l=1}^{N} x_l^{-s/2} e^{-x_l/2} \chi_{x_l > 0}
\prod_{1 \le j < k \le N} | x_k - x_j|,
\end{equation}
where $\mathcal Y_{ (\boldsymbol \gamma - 1 + s)/2}((1)^N):= {\mathcal Y}_{(\boldsymbol \gamma - 1 + s)/2}(x_1,\dots,x_N)|_{x_1=\cdots=x_N = 1}$ and 
\begin{equation}\label{1.12a}
c_N = 2^{-N(N+1)/2}
\prod_{j=1}^N (\Gamma(3/2)/\Gamma(1+j/2)).
\end{equation}
Also $(\boldsymbol \gamma - 1+ s)/2$ denotes the partition $\boldsymbol \gamma$ with each part having $(-1+s)$ added and then divided by 2, with $s=0$ for all parts of $\boldsymbol{\gamma}$ odd, while $s=1$ for all parts of $\boldsymbol{\gamma}$ even.
\end{proposition}

\begin{remark} $ $ 
 \\
1.~(The case $N=2$.) Define the Legendre function of the first kind by
\begin{equation}\label{Lk}
P_\rho(u)  = {1 \over 2 \pi} \int_0^{2 \pi} (u + (u^2 - 1)^{1/2} \cos t )^\rho \, dt =
 {}_2 F_1(\rho+1,-\rho,1;(1-u)/2).
\end{equation}
For $\rho$ a non-negative integer this is a polynomial of degree $\rho$. It is shown in \cite[Eq.~(7.9)]{Ja68} that
\begin{equation}
 \mathcal Y_{(\kappa_1,\kappa_2)}(x_1,x_2) / \mathcal Y_{(\kappa_1,\kappa_2)}(1,1)
= (x_1x_2)^{(\kappa_1+\kappa_2)/2} 
P_{\kappa_1 - \kappa_2}\Big ( {1 \over 2} (x_1 + x_2)/(x_1x_2)^{1/2} \Big ).
\end{equation}
For $N = 2$, and upon use of the hypergeometric form of the Legendre function given in (\ref{Lk}), this extends the functional form (\ref{qW1c+}) to general $\alpha_1, \alpha_2 > -1$. Moreover, this form exhibits the symmetry required of (\ref{qW1c+}) of being unchanged with respect to the interchange of $\alpha_1$ and $\alpha_2$. \\
2.~(The case of a single non-zero part.) Results from \cite{vN41,Ja64} give that 
\begin{equation}
\prod_{l=1}^N (1 - t x_l)^{-1/2} = \sum_{k=0}^\infty {(1/2)_k \over k!}
\mathcal Y_{(k,0,\dots,0)}(x_1,\dots,x_N) t^k.
\end{equation}
This implies an explicit contour integral formula for the zonal polynomials in the case of a single part. We remark that in some recent literature the latter has been referred to as the rank one case
\cite{MP20}. In fact it is possible to analyse the scaled large $N$ limit in this setting, and to obtain phase transition effects for the largest eigenvalue; see the recent review \cite[\S 2.4]{Fo23}.

\end{remark}

\subsection*{Result (4)}
Finally we remark on the evaluation of the characteristic polynomial and the large $N$ limit of the normalised eigenvalue density for $\tilde{W}^{(1)}$.

\begin{proposition}\label{P4x}
    Consider the random matrix $\tilde{W}^{(1)}$ defined as in Proposition \ref{P1x}. For the averaged characteristic polynomial, the functional form (\ref{1.1e+}) again holds true,
\begin{equation}\label{1.1e+2}
\Big \langle \det(x\mathbb I_N - \tilde{W}^{(1)}) \Big \rangle =
\sum_{\nu=0}^N x^\nu \sum_{k=0}^\nu {(-1)^{N-k} \over (\nu - k)! k!}
\prod_{l=1}^N (\alpha_l + k + 1).
\end{equation}
Also, for the specialisation of the parameters (\ref{1.0a+}) in the definition of $\tilde{W}^{(1)}$, 
upon the change of variables $\lambda_l \mapsto (\lambda_l/(N \theta)^\theta)^\theta$, the normalised eigenvalue density of $\tilde{W}^{(1)}$ is given by the Fuss-Catalan density as specified by the moments (\ref{1.1f+}).
\end{proposition}

We turn to the proofs of Results (1)--(4) in the next section. In the case of Proposition \ref{P1.4} this requires making use of the theory of spherical functions. As such, we give an introduction to this theory, both in the complex and the real cases, as a preliminary to the proof.

\section{Proofs of main results}
\subsection{Proof of Proposition \ref{P1}}
\subsubsection*{Proof of Proposition \ref{P1}}
With $Y^{(2)}$ defined by the requirements of the first paragraph of \S \ref{S1.1} we see in particular that the diagonal elements can be taken to be positive real. Doing this, we have that
the joint element PDF of the matrix is
\begin{equation}\label{3.1}
\Big ( \prod_{j=1}^N {2 \over \Gamma(\alpha_j + 1)} y_{jj}^{2 \alpha_j + 1} e^{ - y_{jj}^2} \Big ) {1 \over \pi^{N(N-1)/2}}
\prod_{1 \le j < k \le N} e^{-|y_{jk}|^2}.
\end{equation}

A standard result in the theory of complex Wishart matrices, which can be traced back to Bartlett \cite{Ba33} in the real case, is that for $Y^{(2)}$ an upper triangular matrix with positive real diagonal entries, and complex strictly upper triangular entries, the Jacobian associated with forming the positive definite matrix $W^{(2)} = (Y^{(2)})^\dagger Y^{(2)}$ is specified by
\begin{equation}\label{3.2}
(d W^{(2)}) = 2^N \prod_{j=1}^N y_{jj}^{2(N-j) + 1} (d Y^{(2)});
\end{equation}
see e.g.~\cite[Prop.~3.2.6]{Fo10}. Here the notation $(d Y^{(2)})$ denotes the product of all independent differentials in $Y^{(2)}$ --- real and imaginary parts included separately --- and similarly the meaning of $(d W^{(2)})$. 

It follows that (\ref{3.1}) multiplied by $(d Y^{(2)})$ can be rewritten to involve $W^{(2)}$ and $(d W^{(2)})$ according to
\begin{equation}\label{3.3}
\prod_{j=1}^N {y_{jj}^{2 (\alpha_j - N + j)} \over \Gamma(\alpha_j + 1)}
e^{- {\rm Tr} \, W^{(2)}} (d W^{(2)}).
\end{equation}
It remains to express $\prod_{j=1}^N y_{jj}^{2 (\alpha_j - N + j)}$ in terms of $W^{(2)}$. In this regard, we first make use of the rewrite
$$
\prod_{j=1}^N y_{jj}^{2 (\alpha_j - N + j)} =
\bigg ( \prod_{i=1}^{N-1} 
\prod_{j=1}^{i}
y_{jj}^{2(\alpha_i - \alpha_{i-1} - 1)} \bigg ) \prod_{j=1}^N y_{jj}^{2 \alpha_N}.
$$
The significance of this is that with
$W_i^{(2)}$ defined as the statement of the proposition, the upper triangular form of $Y^{(2)}$ implies
$
\det W_i^{(2)} = \prod_{j=1}^i y_{jj}^2,
$
and consequently
$$
\prod_{j=1}^N y_{jj}^{2 (\alpha_j - N + j)} =
\prod_{i=1}^{N-1} (\det W_i^{(2)})^{\alpha_i - \alpha_{i-1} - 1} 
( \det W^{(2)})^{\alpha_N}.
$$
With this substituted in (\ref{3.3}) we arrive at (\ref{1.5a}). 
\hfill $\square$

\begin{remark} $ $ \\
1.~The proof given above is the complex case analogue of the derivation of (\ref{qW1a}) implied by working in  \cite[proofs of Thms.~4 and 5]{Ve09b}. \\
2.~With $W^{(1)}$ as in Proposition \ref{P1x}, introduce now an invertible, constant $N \times N$ lower triangular matrix with real entries $L$. According to \cite[Th.~5]{Ve09b} the random matrix $X^{(1)}:=
L W^{(1)}L^T$ has joint element PDF
\begin{multline}\label{qW1b}
{1 \over \mathcal N_N^{(1)}(\boldsymbol \alpha)} 
\prod_{i=1}^{N-1} \det ( {}_{N-i}(L L^T))^{(\alpha_i - \alpha_{i+1} - 1)/2}
(\det (L L^T))^{-(\alpha_1 + 1)/2} \\ \times
q^{(1)}_{\boldsymbol \alpha}({X}^{(1)}) e^{- {\rm Tr} \, ({X}^{(1)}(LL^T)^{-1})/2}\chi_{{X}^{(1)}>0},
\end{multline}
where ${}_k A$ denotes the bottom $k \times k$ sub-matrix of $A$. The result (\ref{qW1b}) has an analogue in relation to $W^{(2)}$, which can be derived from Proposition \ref{P1}.

\begin{coro}\label{C1}
Let $L$ be an invertible constant $N \times N$ lower triangular matrix as in (\ref{qW1b}), but now allow complex elements. Let $W^{(2)}$ be as in the first paragraph of this subsection. The random matrix $X^{(2)}:= L W^{(2)}L^\dagger$ has joint element PDF
\begin{equation}\label{qW1c}
{1 \over \mathcal N_N^{(2)}(\boldsymbol \alpha)} 
\prod_{i=1}^{N-1} \det ( {}_{N-i}(L L^\dagger))^{\alpha_i - \alpha_{i+1} - 1}
(\det (L L^\dagger))^{-(\alpha_1 + 1)}
q^{(2)}_{\boldsymbol \alpha}({X}^{(2)}) e^{- {\rm Tr} \, ({X}^{(2)}(LL^\dagger)^{-1})}.
\end{equation}
\end{coro}

\subsubsection*{Proof of Corollary \ref{C1}}
We now turn our attention to the proof of Corollary \ref{C1}.
To begin we substitute
\begin{equation}\label{WX}
W^{(2)} = L^{-1} X^{(2)} (L^\dagger)^{-1}
\end{equation}
in (\ref{1.5a}), and make use of the readily verified factorisation
\begin{equation}\label{WX1}
q_{\mathbf \alpha}^{(2)}(L^{-1} X^{(2)} (L^\dagger)^{-1}) =
q_{\mathbf \alpha}^{(2)}((L^\dagger L)^{-1}) q_{\mathbf \alpha}^{(2)}(X^{(2)}),
\end{equation}
which is reliant of $L$ being lower triangular.
Also, it is a standard result that for the change of variables (\ref{WX}) (see e.g.~\cite[Exercise 1.3 q.2 with $\beta =2$]{Fo10})
\begin{equation}\label{WX2}
(d W^{(2)}) = ( \det ((L^\dagger L)^{-1})^N  (d X^{(2)} ).
\end{equation}
Now forming the product (\ref{WX1}) and (\ref{WX2}), and using the identity
$$
q_{\mathbf \alpha}^{(2)}((L^\dagger L)^{-1}) ( \det ((L^\dagger L)^{-1})^N =
\prod_{i=1}^{N-1} \det ( {}_{N-i}(L L^\dagger))^{\alpha_i - \alpha_{i+1} - 1}
(\det (L L^\dagger))^{-(\alpha_1 + 1)},
$$
where again the fact the requirement that $L$ be lower triangular is essential,
the result (\ref{qW1c}) follows.
\hfill $\square$
\end{remark}

\subsection{Proof of Proposition \ref{P1y+}}
Let $C^{(2)}$ be given by (\ref{C1r}), and as an abbreviation set $(W^{(2)})_{jj} = t_j > 0$. Writing $W^{(2)}$ in terms of $C^{(2)}$ and $\{t_j\}$ is a change of variables, for which we have
$$
(d W^{(2)}) = \prod_{j=1}^N t_j^{N-1} dt_j (d C^{(2)}).
$$
Also, upon making this substitution  in (\ref{1.17b}) we obtain
$$
q_{\boldsymbol \alpha}^{(2)}(W^{(2)}) =
\prod_{l=1}^N t_l^{\alpha_l + (l-1) - (N-1)} q_{\boldsymbol \alpha}^{(2)}(C^{(2)}).
$$
Hence in the variables of $C^{(2)}$ and $\{t_j\}$ the PDF (\ref{1.5a}) reads
$$
{1 \over \mathcal N_N^{(2)}(\boldsymbol \alpha)} q^{(2)}_{\boldsymbol \alpha}({C}^{(2)}) \prod_{l=1}^N t_l^{\alpha_l + l - 1} e^{-t_l} \chi_{t_l > 0} \chi_{C_{ii}^{(2)} \, (i=1,\dots,N), \, C^{(2}) >0}.
$$
The result (\ref{qW1C}) now follows by integrating over $\{t_j\}$.

\subsection{Proof of Proposition \ref{P1.4}}
\subsubsection*{Introductory remarks --- complex spherical functions}
It is instructive to consider the complex case first.
A tie in with the theory of spherical functions shows itself when seeking the passage way from the joint element PDF of Proposition \ref{P1} to the corresponding eigenvalue PDF (\ref{1.0}).
Let $d \mu_U$ denote the normalised Haar measure on the space of $N \times N$ complex unitary matrices.
We have the decomposition of measure (see e.g.~\cite[Eq.~(1.27)]{Fo10})
$$
(d W^{(2)}) \propto \prod_{1 \le j < k \le N} (x_k - x_j)^2 \, d \mu_U,
$$
associated with the diagonalisation formula $W^{(2)} = U \Lambda U^\dagger$, where $U$ denotes the unitary matrix of eigenvectors and $\Lambda$ the diagonal matrix of eigenvalues. Substituting this in the element PDF formula of Proposition \ref{P1} and integrating over the eigenvectors, to obtain (\ref{1.0}) it must be that
\begin{equation}\label{qC1}
\int q^{(2)}_{\boldsymbol \alpha}({U \Lambda U^\dagger})
\, d \mu_U =
\prod_{l=1}^N {\Gamma(l) \over \Gamma(\alpha_l+1)}
{ \det [x_j^{\alpha_k}]_{j,k=1}^N \over 
\prod_{1 \le j < k \le N} (\alpha_j - \alpha_k) (x_j - x_k)}.
\end{equation}
Indeed (\ref{qC1}) is a matrix integral evaluation due to Gelfand and Na\u{\i}mark \cite{GN57} obtained in the context of the study of harmonic analysis on symmetric spaces, and is an example of a spherical function \cite{He84,Ma95}. Recent applications of this complex spherical function in random matrix theory can be found in 
\cite{KK16,ZKF21,KLZF23,Me22,MN22}

A striking feature of (\ref{qC1}) is that the corresponding group integral over $U(N)$ is in fact symmetric in $\{ \alpha_k \}$. 
This can be viewed as a consequence of the eigenvalue PDF of $W^{(2)}$ being symmetric in $\{ \alpha_k \}$, which can be established independently.

\begin{lemma}\label{L1}
The characteristic polynomial of $W^{(2)}$ is unchanged in distribution by interchanges of the parameters $\{\alpha_k \}$, and hence so is the eigenvalue PDF.
\end{lemma}

\begin{proof}
Consider
the characteristic polynomial
$\det ( W^{(2)} - z \mathbb I )$, where $W^{(2)} = (Y^{(2)})^\dagger Y^{(2)}$ For $P_1, P_2$ permutation matrices, this is unchanged by replacing $Y^{(2)}$ by $P_1 Y^{(2)} P_2$, and also interchanging the order of this transformation of $Y^{(2)}$ with $(Y^{(2)})^\dagger$. On the other hand such transformations can be used to permute the location of any two $\alpha_k$'s in $Y^{(2)}$, while permuting the strictly upper triangular entries among themselves. The latter are all identically distributed, so we must have that the distribution of the characteristic polynomial is symmetric in $\{\alpha_k \}$.
\end{proof}

\subsubsection*{Introductory remarks --- real spherical functions}

To specify the group integral relevant to Proposition \ref{P1.4},
let $d \mu_R$ denote the normalised Haar measure on the space of $N \times N$ real orthogonal matrices. Diagonalising $W^{(1)}$ according to $W^{(1)} = R \Lambda R^T$, where $R$ denotes the real orthogonal matrix of eigenvectors and $\Lambda$ denotes the diagonal matrix of eigenvalues $\{ x_l \}$, we have for the corresponding decomposition of measure
\begin{equation}\label{2.g}
(d W^{(1)}) = k_N \prod_{1 \le j < k \le N} (x_k - x_j) \, d \mu_R,
\end{equation}
for some proportionality $k_N$;
see e.g.~\cite[Eq.~(1.11)]{Fo10}. 
Hence to go from the element PDF   (\ref{qW1a}) to the eigenvalue PDF we must specify the function of the eigenvalues specified by the group integral $\int q^{(1)}_{\boldsymbol \alpha}({R \Lambda R^T})
\, d \mu_R$.

It turns out that this group integral has a number of structured properties. One is that
\begin{equation}\label{po}
\int q^{(1)}_{\boldsymbol \alpha}({\Sigma^{1/2} R \Lambda R^T}\Sigma^{1/2}) \, d \mu_R =
\bigg ( \int q^{(1)}_{\boldsymbol \alpha}({ R \Lambda R^T}) \, d \mu_R \bigg )
\bigg ( \int q^{(1)}_{\boldsymbol \alpha}({ R \Sigma R^T})
\, d \mu_R \bigg );
\end{equation}
for a proof see \cite[Pg.~103]{Me22}. As such, the group integral is identified as a (multiplicative) spherical integral; see e.g.~\cite{Ma95} for generalities. Also, in keeping with the real version of Lemma \ref{L1}, the group integral is symmetric in  $\{\alpha_l \}$.
For general real $\alpha_k$ this spherical function can be identified as an example of a so-called Heckman-Opdam hypergeometric function \cite{HO87,Su16}. However this fact is of limited practical utility, as in general there are no associated structured functional forms.

 An exception is
 when the partition $\boldsymbol \gamma$ defined by
(\ref{gat}) and surrounding text is introduced in favour of $\{ \alpha_k\}$. Then, it has been known since the work of James \cite{Ja64} (see also \cite{KM84}) that subject to the further requirement that the exponents in (\ref{q1W}) are all integer, the spherical function is in fact a polynomial in the eigenvalues, and can be alternatively identified (up to proportionality) as the zonal polynomial of mathematical statistics, $\mathcal Y_{\boldsymbol\kappa} (x_1,\dots,x_N)$ say. This fact is the starting point for our derivation of (\ref{qW1c+}). 

The practical utility of identifying (a restriction of) the spherical function in terms of the zonal polynomials rather than a Heckman-Opdam hypergeometric function
is that the former
admit other characterisations which can be used for computations (although not an explicit general formula as implied by (\ref{qC1}) in the complex case). Thus expanded in terms of the monomial symmetric polynomials
$\{m_\nu \}$, they have a triangular structure
 \cite[Ch.~7]{Mu82}
\begin{equation}\label{s1}
\mathcal Y_{\boldsymbol\kappa}(x_1,\dots,x_N) = c_{\boldsymbol\kappa,\boldsymbol\kappa} m_{\kappa}(x_1,\dots,x_N) + \sum_{\boldsymbol\mu < \boldsymbol\kappa} c_{\boldsymbol\kappa,\boldsymbol\mu} m_{\boldsymbol\mu}(x_1,\dots,x_N),
\end{equation}
for a particular non-zero constant $c_{\boldsymbol\kappa,\boldsymbol\kappa}$ (independent of $N$ and $\{x_i\}$).  The notation $\boldsymbol\mu < \boldsymbol\kappa$ denotes the partial order on partitions of the same size $|\boldsymbol\mu|$ and $|\boldsymbol\kappa|$ (the size of a partition is the sum of the parts) defined by  the requirement that $\sum_{i=1}^p \mu_i \le \sum_{i=1}^p \kappa_i$ for each $p=1,\dots,N$. Up to the normalisation $c_{\boldsymbol\kappa,\boldsymbol\kappa}$, the coefficients in (\ref{s1}) can be explicitly computed by the property of the zonal polynomials
that they are the  eigenfunctions of the differential operator
$$
\sum_{j=1}^N x_j^2 {\partial^2 \over \partial x_j^2} + \sum_{i,j=1 \atop j \ne i} {x_i^2 \over x_i - x_j} {\partial \over \partial x_i}.
$$
The choice of normalisation used in mathematical statistics specifies $c_{\boldsymbol\kappa,\boldsymbol\kappa}$ so that
$$
\sum_{|\kappa| = k} \mathcal Y_{\boldsymbol\kappa}(x_1,\dots,x_N) = (x_1+\cdots+x_N)^k.
$$
We refer to the recent work \cite{LK20} for more on the explicit generation of zonal polynomials from this characterisation, and also alternative characterisations.


\subsubsection*{Workings of the proof}
The symmetry property in $\{\alpha_k\}$ of the group integrals in (\ref{po}) as implied by the real version 
of Lemma \ref{L1}
 tells us that we have
\begin{equation}\label{2.8}
\int {q}^{(1)}_{\boldsymbol \alpha}(R \Lambda R^T) \,  d \mu_R =
\int {q}^{(1)}_{\overleftarrow {\boldsymbol \alpha}}(R \Lambda R^T) \,  d \mu_R.
\end{equation}
Here $\overleftarrow {\boldsymbol \alpha} = (\alpha_N,\dots, \alpha_1)$,
or equivalently in terms of $\boldsymbol \gamma$ introduced in the sentence including (\ref{gat}), 
\begin{equation}\label{2.9}
\overleftarrow {\boldsymbol \alpha} = (\gamma_1 + N - 1, \gamma_2 + N - 2,\dots,\gamma_N).
\end{equation}
Now define
\begin{equation}\label{2.9a}
    \tilde{q}^{(1)}_{\boldsymbol \kappa}(X) := 
\prod_{i=1}^{N-1} ( \det ({X}_i)^{\kappa_i - \kappa_{i+1}}
( \det {X} )^{\kappa_N}.
\end{equation}
Comparing this with the definition of ${q}^{(1)}_{\boldsymbol \alpha}(X)$ in (\ref{q1W}) and taking into consideration (\ref{2.9}) and (\ref{2.8}) shows
\begin{equation}\label{2.10}
\int {q}^{(1)}_{\boldsymbol \alpha}(R \Lambda R^T) \,  d \mu_R =
\int \tilde{q}^{(1)}_{({\boldsymbol \gamma}-1)/2}(R \Lambda R^T) \,  d \mu_R,
\end{equation}
where $({\boldsymbol \gamma}-1)/2$ denotes the array in (\ref{gat}) with each part having one subtracted and then divided by 2.

On the other hand, the theory of the introductory remarks of this subsection tells us that for $\kappa$ a partition
\begin{equation}
{\mathcal Y_{\boldsymbol\kappa} (x_1,\dots,x_N) \over \mathcal Y_{\boldsymbol\kappa} ((1)^N))} =
\int \tilde{q}^{(1)}_{\boldsymbol \kappa}(R \Lambda R^T) \,  d \mu_R, 
\end{equation}
where $Y_{\boldsymbol\kappa} ((1)^N))$ denotes $Y_{\boldsymbol\kappa} (x_1,\dots,x_N)$ evaluated at $x_1 = \cdots = x_N = 1$ as used in (\ref{qW1c+}). Hence if all the parts of $(\boldsymbol\gamma-1)/2$ are an integer, and thus all parts of $\boldsymbol\gamma$ are odd we have
\begin{equation}\label{2.a}
\int {q}^{(1)}_{\boldsymbol \alpha}(R \Lambda R^T) \,  d \mu_R = {\mathcal Y_{\boldsymbol(\gamma-1)/2} (x_1,\dots,x_N) \over \mathcal Y_{\boldsymbol(\gamma-1)/2} ((1)^N))}.
\end{equation}
Noting too from the definition (\ref{2.9a}) that for any constant $c$ added to each part of $\kappa$,
$$
\tilde{q}^{(1)}_{\boldsymbol \kappa + c }(X) = (\det X)^{c}
\tilde{q}^{(1)}_{\boldsymbol \kappa }(X),
$$
in addition to (\ref{2.a}) we have
\begin{equation}\label{2.b}
\int {q}^{(1)}_{\boldsymbol \alpha}(R \Lambda R^T) \,  d \mu_R = \prod_{l=1}^N x_l^{-1/2} {\mathcal Y_{\boldsymbol\gamma/2} (x_1,\dots,x_N) \over \mathcal Y_{\boldsymbol\gamma/2} ((1)^N))},
\end{equation}
valid for all parts of $\boldsymbol\gamma$ are even. Combining with (\ref{2.g}), (\ref{2.a}) and (\ref{2.b}) establishes (\ref{qW1c+}) up to an $N$-dependent proportionality $k_N$ in (\ref{2.g}). The particular choice $\gamma_{N+1-k} - 1 = 0$ ($k=1,\dots,N$) reduces  (\ref{2.a}) to unity. From the definition above (\ref{gat}) this choice is equivalent to taking $\alpha_k = k $ ($k=1,\dots,N$). Hence changing variables to the eigenvalues and eigenvectors in the element PDF (\ref{qW1a}) and integrating over both, the proportionality $k_N$ in (\ref{2.g}) must be such that
$$
{k_N \over \mathcal N_N^{(1)}(\boldsymbol \alpha) |_{\{\alpha_k = k \}}} \int_{0}^\infty d x_1 \cdots \int_{0}^\infty d x_N \, \prod_{l=1}^N e^{- x_l/2} \prod_{1 \le j < k \le N} | x_k - x_j|  = 1.
$$
The explicit evaluation of the multiple integral is well known \cite[Prop.~4.7.3 with $\beta =1$, $a=0$]{Fo10} which specifies $k_N$, and completes the derivation of (\ref{qW1c+}). \hfill $\square$

\begin{remark}\label{R1.3} 
 The paper \cite{DG14}, Theorem 2 with $\beta = 1$ and $\Sigma = I$ claims a formula for the eigenvalue PDF of (\ref{qW1a}). However this is incorrect due to its reliance on the erroneous Eq.~(4) of the same paper.
 \end{remark}

\subsection{Proof of Proposition \ref{P4x}}
In relation to the averaged characteristic polynomial, one considers 
\begin{equation}\label{2.k}
\det \bigg ( x \mathbb I_{n+N} -
\begin{bmatrix} 0_n & \tilde{Y}^{(2)} \\ (\tilde{Y}^{(2)})^\dagger & 0_N \end{bmatrix} \bigg ) = x^{n-N} \prod_{l=1}^N (x^2 - \lambda_l^2 ),
\end{equation}
where $\{ \lambda_l^2\}$ are the eigenvalues of $(\tilde{Y}^{(2)})^\dagger \tilde{Y}^{(2)}$. Consequently
$$
\Big \langle \det(x\mathbb I_N - \tilde{W}^{(2)}) \Big \rangle = \prod_{l=1}^N (x - \lambda_l^2 ).
$$
The characteristic polynomial on the LHS of (\ref{2.k}), in the case that all the entries are identically distributed with mean zero and (complex) variance $\sigma^2$ has been studied previously \cite[Prop.~12]{FG06}. The key point is that the averaged characteristic polynomial is independent of other statistical features of the entries. Inspection of the proof shows that this remains true in the case that some entries in $\tilde{Y}^{(2)}$ are set equal to zero. In particular, in this setting the choice of standard real Gaussian entries, and standard complex Gaussian entries, have the same averaged characteristic polynomial. The result (\ref{1.1e+2}) now follows from knowledge that the RHS applies in the case of standard complex Gaussian entries for the nonzero entries \cite[Corollary 3.7]{FI18}.

The independence of the averaged characteristic polynomial on statistical properties of $\tilde{Y}^{(2)}$ beyond the first two moments of the entries is a form of universality. If one considers instead the normalised scaled spectral density, it is again the case that universality results hold true, although unlike for the averaged characteristic polynomial they require the large $N$ limit for their validity. The relevant theory is contained in \cite[Th.~3.7 and pgs.~47--48]{BS10} and \cite[Th.~2.1.21]{AGZ09}. It is also summarised in \cite[Paragraph below (40)]{CLM23}. This theory gives that the limiting normalised scaled spectral density in the case of standard real Gaussian entries of the nonzero entries of $\tilde{Y}^{(2)}$ is the same as that for the case of standard complex Gaussian entries. Since the latter is specified in terms of the Fuss-Catalan density, so is the former.

\subsection*{Acknowledgements}
	This research is part of the program of study supported
	by the Australian Research Council Discovery Project grant DP210102887.


\begin{thebibliography}{10}


  \bibitem{AGZ09}
G.W. Anderson, A.~Guionnet, and O.~Zeitouni, \emph{An introduction to random
  matrices}, Cambridge University Press, Cambridge, 2009.

\bibitem{Ba33}
M.S. Bartlett, \emph{On the theory of statistical regression}, Proc. R. Soc. Edinb. \textbf{53} (1933) 260--283.

\bibitem{Be56}
R.~Bellman, \emph{A generalization of some integral identities due to Ingham and Siegel}, Duke Math.
J., 23 (1956), 571--577.

\bibitem{BS10}
Z.~Bai and J.W.~Silverstein, \emph{Spectral analysis of large dimensional random matrices}, 2nd edition,
Springer, 2010.

\bibitem{BBCC11}
T. Banica, S. T. Belinschi, M. Capitaine, and B. Collins, \emph{Free Bessel laws}, Canad. J. Math.
\textbf{63} (2011), 3--37.


\bibitem{Bo98}
A.~Borodin, \emph{Biorthogonal ensembles}, Nucl. Phys. B \textbf{536} (1998),
  704--732.


  \bibitem{BF22a}
   S.-S.~Byun and P.J.~Forrester,   \emph{Progress on the study of the Ginibre ensembles I: G{\SMALL in}UE},
   arXiv:2211.16223.

\bibitem{CC21}
   C. Charlier and T. Claeys, \emph{Global rigidity and exponential moments for soft and hard edge
point processes}, Prob. Math. Phys. \textbf{2} (2021), 363--417.


\bibitem{CLM21}
   C. Charlier, J. Lenells, and J. Mauersberger, \emph{Higher order large gap asymptotics at the
hard edge for Muttalib-Borodin ensembles}, Comm. Math. Phys. \textbf{384} (2021), 829--907.


  \bibitem{Ch14}
D.~Cheliotis, \emph{Triangular random matrices and biothogonal ensembles},
Statist. Prob. Letters \textbf{134} (2018), 36--44.

\bibitem{CGS19}
T. Claeys, M. Girotti, and D. Stivigny, \emph{Large gap asymptotics at the hard edge for product
random matrices and Muttalib-Borodin ensembles}, Int. Math. Res. Not.  \textbf{2019} (2019), 2800--2847.


\bibitem{CR14}
T. Claeys, and S. Romano,
\emph{Biorthogonal ensembles with two-particle interactions.}
Nonlinearity 27 (2014) 2419.


\bibitem{CLM23}
F.D.~Cunden, M.~Ligab\`o and T.~Monni,
\emph{Random matrics associated to Young diagrams}, arXiv:2301.13555.

\bibitem{DG14}
J.A.~D\'iaz-Garc\'ia,
\emph{On Riesz distribution}, Metrika \textbf{77} (2014), 469--481.

\bibitem{GKLS22}
E.~Ghorbel, K.~Kammoun,  M.~Louati and A.~Sallem, \emph{Estimation of the parameters of a Wishart extension on symmetric matrices}, J. Korean Stat. Soc. \textbf{51} (2022), 1071--1089. 

\bibitem{GL19}
E.~Ghorbel and M.~Louati,  \emph{The multiparameter t'distribution}, Filomat \textbf{33} (2019), 4137--4150.


\bibitem{FK94}
J.~Faraut and A.~Kor\'anyi,  \emph{Analysis on symmetric cones}, Oxford Mathematical
Monographs, Clarendon Press, 1994.


    \bibitem{Fo10}
P.J. Forrester, \emph{Log-gases and random matrices}, Princeton University Press,
  Princeton, NJ, 2010.

 \bibitem{Fo23}
P.J. Forrester, \emph{Rank 1 perturbations in random matrix theory -- a review of exact results},
Random Matrices: Theory and Applications, to appear, arXiv:2201.00324. 

\bibitem{FG06}
P.J. Forrester and A.~Gamburd, \emph{Counting formulas associated with some
  random matrix ensembles}, J. Combin. Th. Series A \textbf{113} (2006),
  934--951.

\bibitem{FI18}
  P.J. Forrester and J.R. Ipsen, \emph{Selberg integral theory and Muttalib-Borodin ensembles}, Adv. Appl.
Math. \textbf{95} (2018), 152--176.

\bibitem{FL22}
P.J. Forrester and S.-H. Li, \emph{Rate of convergence at the hard edge for various
P \'olya ensembles of positive definite matrices}, Integral Transforms Spec. Funct. \textbf{33} (2022), 466--484.

\bibitem{FL15}
P. J. Forrester, and D.-Z. Liu,
\emph{Raney distributions and random matrix theory.}
J. Stat. Phys. \textbf{158} (2015) 1051.

  \bibitem{FW17} 
  P.J. Forrester and D. Wang,  \emph{Muttalib-Borodin ensembles in random matrix theory --- realisations and correlation functions},  Elec. J. Probab. \textbf{22} (2017), 54.

 \bibitem{FZ20}
P.J. Forrester and J.~Zhang,  
\emph{Parametrising correlation matrices}, Journal of Multi-
variate Analysis \textbf{178} (2020), 104619.  


  \bibitem{GN57}
I.M. Gelfand and M.A. Na\u{\i}mark, \emph{Unit\"{a}re Darstellungen der klassischen Gruppen,} Akademie-Verlag, Berlin (1957), translated from Russian: Trudy Mat. Inst. Steklov. 36 (1950), 288.

\bibitem{GN00}
 A.K. Gupta and D.K. Nagar, \emph{Matrix variate distributions}, Chapman and Hall/CRC, Boca Raton, FL, 2000.

\bibitem{Ha21}
 A Hassairi, \emph{Riesz probability distributions}, De Gruyter Series in Probability and Stochastics, 2021.

\bibitem{HO87}
G. J. Heckman and E. M. Opdam, \emph{Root systems and hypergeometric functions. I}, Compositio
Mathematica, \textbf{64} (1987), 329--352.

\bibitem{He84}
S. Helgason, \emph{Groups and Geometric Analysis. Integral geometry, invariant differential
operators, and spherical functions. Corrected reprint of the 1984 original}, 1984.

\bibitem{In33}
A.E. Ingham, \emph{An integral which occurs in statistics}, Proc. Cambridge Philos. Soc. \textbf{29} (1933),
271--276.

\bibitem{Ja60}
A.T.~James, \emph{The distribution of the latent roots of the covariance matrix},
Ann. Math. Stat. \textbf{31} (1960), 151--158.

\bibitem{Ja61}
A.T.~James, \emph{Zonal polynomials of the real positive definite symmetric matrices}, Ann. Math.  \textbf{74} (1961), 456--469.

\bibitem{Ja64}
A.T.~James, \emph{Distributions of matrix variate and latent roots derived from
  normal samples}, Ann. Math. Statist. \textbf{35} (1964), 475--501.

\bibitem{Ja68}
A.T. ~James, \emph{Calculation of zonal polynomial coefficients by use of the Laplace-Beltrami operator}, Annals Math. Stat. \textbf{39} (1968), 1711--1718.

\bibitem{LK20}
L. Jiu and C. Koutschan, \emph{Calculation and properties of zonal polynomials}, Mathematics in Computer Science. \textbf{14} (2020), 623--640.

\bibitem{KK16}
M. Kieburg, and H. K\"osters,
\emph{Exact Relation between Singular Value and Eigenvalue Statistics}. 
Random Matrices: Theor. Appl. \textbf{5} (2016) 1650015.

\bibitem{KLZF23}
M. Kieburg, S.-L.~Li, J.~Zhang and P.J.~Forrester, \emph{Cyclic P\'olya ensembles on the unitary matrices and their spectral statistics}, Constr. Approx. \textbf{57} (2023), 1063--1108.



\bibitem{KM84}
  H.B. Kushner and M.~Meisner, \emph{Formulas for zonal polynomials}, J. Mult. Analysis,
\textbf{14} (1984), 336--347.

 \bibitem{La18}
 G.~Lambert, \emph{Limit theorems for biorthogonal ensembles and related combinatorial identities},
 Adv.~Math. \textbf{329}, (2018) 590--648.  



\bibitem{Ma95}
I.G. Macdonald, \emph{Hall polynomials and symmetric functions}, 2nd ed., Oxford
  University Press, Oxford, 1995.


\bibitem{MN22}
  C. McSwiggen and J. Novak, \emph{Majorization and spherical functions},
  International Mathematics Research Notices, \textbf{2022} (2022), 3977--4000.


  \bibitem{Me22}
  P.~Mergny, \emph{Spherical integral and their applications to random matrix theory}, PhD thesis, Universit\'e Paris-Saclay, 2022. 

 \bibitem{MP20}
 P.~Mergny and M.~Potters,
 \emph{Asymptotic behavior of the multiplicative counterpart of the
Harish-Chandra integral and the S-transform}, arXiv:2007.09421.


\bibitem{Ml10}
W.~M{\l}otkowski, \emph{Fuss-catalan numbers in noncommutative probability},
  Documenta Mathematica \textbf{15} (2110), 939--955.

\bibitem{Mo21}
L.D. Molag, \emph{The local universality of Muttalib-Borodin ensembles when the parameter $\theta$
is the reciprocal of an integer}, Nonlinearity, \textbf{34} (2021) 3485--3564.

  \bibitem{Mu82}
R.J. Muirhead, \emph{Aspects of multivariate statistical theory}, Wiley, New
  York, 1982.

  \bibitem{Mu95}
K.A. Muttalib, \emph{Random matrix models with additional interactions}, J.
  Phys. A \textbf{28} (1995), L159--164.

 \bibitem{Ri49}
  M. Riesz,
\emph{L'intégrale de Riemann-Liouville et le problème de Cauchy}.
Acta Mathematica \textbf{81} (1949), 1--223.

 \bibitem{Si35}
  C.L. Siegel, \emph{\"Uber die analytische theorie der quadratischen formen} Ann. Math. \textbf{36}
(1935), 527--606. 

\bibitem{Su16}
  Y. Sun, \emph{A new integral formula for Heckman-Opdam hypergeometric functions}, Adv. in Math. \textbf{289} (2016),
1157--1204.

\bibitem{Ve09a}
E.~Velva, \emph{Testing a normal covariance matrix for small samples with monotone missing data}, Applied Mathematical Sciences \textbf{3} (2009), 2671--2679.

\bibitem{Ve09b}
E.~Veleva, \emph{Stochastic representations of the Bellman gamma distribution}, In Proceeding of international conference on theory and
applications in mathematics and informatics, ICTAMI 2009. Alba Iulia, Romania (2009), 463--474.

\bibitem{Ve11}
E.~Veleva, \emph{Properties of the Bellman gamma distribution}, Pliska Studia Mathematica Bulgar. \textbf{20} (2011), 221--232.

\bibitem{vN41}
J.~von Neumann, \emph{Distribution of the ratio of the mean square successive difference to
the variance}, Annals Math. Stat. \textbf{12} (1941), 367--395.


\bibitem{WZ22}
D.~Wang and L.~Zhang, \emph{A vector Riemann-Hilbert approach to the Muttalib-Borodin
ensembles}, J. Funct. Anal. \textbf{282} (2022), 109380.

\bibitem{Wi28}
J. Wishart, \emph{The generalized product moment distribution in samples from a normal multivariate population},
Biometrika \textbf{20A} (1928), 32--43.

  \bibitem{Zh15}
L. Zhang, 
\emph{Local universality in biorthogonal Laguerre ensembles.} 
J. Stat. Phys. \textbf{161}, (2015) 688--711.

\bibitem{ZKF21}
J.~Zhang, M.~Kieburg and P.J.~Forrester,
\emph{Harmonic analysis for rank-1 randomised Horn problems}, Lett. Math. Phys. \textbf{111} (2021), 98.










\end{thebibliography}
\nopagebreak

\providecommand{\bysame}{\leavevmode\hbox to3em{\hrulefill}\thinspace}
\providecommand{\MR}{\relax\ifhmode\unskip\space\fi MR }
\providecommand{\MRhref}[2]{%
  \href{http://www.ams.org/mathscinet-getitem?mr=#1}{#2}
}
\providecommand{\href}[2]{#2}

\end{document}